\RequirePackage[l2tabu, orthodox]{nag} 
\documentclass[UTF8]{article}
\usepackage[marginratio=1:1,height=584pt,width=400pt,tmargin=117pt]{geometry} 
\usepackage[english]{babel}

\usepackage{graphicx}
\usepackage{float}  
\usepackage{booktabs}
\usepackage{tikz-cd}
\usepackage{tikz}
\usetikzlibrary{decorations.markings,positioning}
\usetikzlibrary{shadows.blur}
\usepackage{standalone}
\usepackage[font=small]{caption}
\usepackage{subcaption}

\usepackage{footmisc}

\usepackage[page,toc,titletoc,title]{appendix}

\usepackage[backend=biber, 
    style=numeric-comp,
   sorting=none]{biblatex}
\usepackage{biblatex} 

\usepackage{slashed}
\usepackage{faktor}
\usepackage{stackrel}
\usepackage{enumitem}
\usepackage{bbold}
\usepackage{csquotes}
\usepackage{ dsfont }
\usepackage{listings}
\usepackage{microtype} 
\usepackage{blindtext}
\usepackage{enumitem}

\usepackage{amsthm}
\usepackage{amsmath,amsfonts,amssymb,amsthm,epstopdf,titling,url,array,bm}
\usepackage{xfrac}
\usepackage{braket}
\usepackage{authblk}
\usepackage{hyperref}
\usepackage{cleveref} 

\usepackage{todonotes}

\newtheorem{theorem}{Theorem}
\theoremstyle{plain}

\newtheorem{lemma}{Lemma}

\theoremstyle{definition}

\newtheorem{assumption}{Assumption}

\usepackage{xcolor}
\usepackage{todonotes}

\usepackage{soul}

\title{Comment on a no-go theorem for $\psi$-ontic models}
\author{Laurens Walleghem, Shashaank Khanna, Rutvij Bhavsar \\ 
\href{mailto:laurens.walleghem@york.ac.uk}{laurens.walleghem@york.ac.uk}, \href{mailto:shashaank.khanna@york.ac.uk}{shashaank.khanna@york.ac.uk}, \href{mailto:rutvij@kaist.ac.kr}{rutvij@kaist.ac.kr}}
\affil{Department of Mathematics,  University of York, Heslington, York, YO10 5DD, United Kingdom}
\date{\today}

\begin{document}

\maketitle

\begin{abstract}
    In a recent paper [Carcassi, Oldofredi and Aidala, Found Phys 54, 14 (2024)] it is claimed that the whole Harrigan--Spekkens framework of ontological models is inconsistent with quantum theory. They show this by showing that all pure quantum states in $\psi$-ontic models must be orthogonal. In this note, we identify some crucial mistakes in their argument to the extent that the main claim is incorrect.
\end{abstract}

\section{Introduction}
In this short note, we consider the result of \cite{carcassi2022reality}, where grand claims on the nature of the quantum state are made. It is claimed that the Harrigan-Spekkens \cite{harrigan2010einstein} categorisation of ontic vs epistemic states itself is problematic. The authors of \cite{carcassi2022reality} claim that for $\psi$-ontic models all pure states in quantum mechanics must be orthogonal, in stark violation of quantum mechanics. Hence, they conclude that their result shows that even $\psi$-ontic models cannot reproduce Quantum Mechanics. An immediate contradiction to their result is the existence of the de Broglie- Bohm theory which is $\psi$-ontic and reproduces Quantum Mechanics. Together with their result and the PBR theorem \cite{pusey2012reality}, which suggests that the wavefunction in quantum mechanics is ontic, the authors of \cite{carcassi2022reality}  conclude that it is the Harrigans-Spekkens framework itself that is wrong. 
As we will argue, three points weaken their argument to the point that their conclusion does not hold together. Namely, \begin{itemize}
    \item it is assumed that the entropy of the quantum state equals the entropy of the underlying ontological distribution, which need not necessarily be the case,
    \item even if we do assume that the entropy of the quantum state should be the same as the entropy of the underlying ontological distribution (which as mentioned is meaningless), their result would still restrict to $\psi$-complete models and does not capture \emph{all} $\psi$-ontic models, 
    \item the claim (and its proof) of \cite{carcassi2022reality} that `no ontological model can reproduce quantum theory' is wrong. The research paradigm of generalised noncontextuality \cite{spekkens2005contextuality,spekkens2014status,schmid2018all,schmid2020structure,schmid2021characterization,mazurek2016experimental,kunjwal2019beyond,chaturvedi2021characterising,schmid2023addressing} and related establishes that rather `any ontological model for quantum theory needs fine-tuning'.
\end{itemize}

\section{Recap of the argument of \cite{carcassi2022reality}}

We shortly recap the argument of \cite{carcassi2022reality} here and state the assumptions upon which we comment in the next section. 

Two non-orthogonal quantum states $|\psi\rangle, |\phi\rangle$ are considered, and the entropy of their mixture $\rho = \frac{1}{2} | \psi \rangle \langle \psi | + \frac{1}{2} |\phi \rangle \langle \phi |$ is calculated in two ways, namely (i) using the von Neumann entropy for quantum states and (ii) using the ontological $\psi$-ontic representation of those quantum states and assuming convexity-preserving. Using the fact that $\rho$ can be diagonalised as $\rho=\frac{1+|\langle\psi \mid \phi\rangle|}{2}|+\rangle\langle+|+\frac{1-|\langle\psi \mid \phi\rangle|}{2}|-\rangle\langle-|$ for some basis $| \pm \rangle$ the von Neumann entropy is given by  \begin{equation} \label{eq:vonNeumann}
\begin{split}
    H(\rho) &= - \operatorname{Tr} \big[ \rho \log \rho \big] =H\left(\frac{1+|\langle\psi \mid \phi\rangle|}{2}, \frac{1-|\langle\psi \mid \phi\rangle|}{2}\right) \\ 
    &=-\frac{1+\sqrt{p}}{2} \log \frac{1+\sqrt{p}}{2}-\frac{1-\sqrt{p}}{2} \log \frac{1-\sqrt{p}}{2},
\end{split}
\end{equation} with $p = |\langle \psi | \phi \rangle|^2$. On the other hand, consider an ontological model for quantum theory, that is $\psi$-ontic such that $\mu(\lambda|\psi)$ and $\mu(\lambda|\phi)$ have nonoverlapping support $U_\psi$ and $U_\phi$. For ease of notation, we define $\mu_{P_{\sigma}}(\lambda) := \mu(\lambda|P_\sigma)$ for a density operator $\sigma$. Furthermore, the following is assumed, as stated in eq. (10) in \cite{carcassi2022reality}. 
\begin{assumption} \label{ass:convex-preserving}
    The ontological model preserves convex combinations of preparations, i.e. if density operator $\rho$ with preparation $P_\rho$ is obtained by sampling between preparing $P_\psi$ and $P_\phi$ with equal probability, then their ontological representations satisfy \begin{equation}
    \mu_{P_{\rho}} = \frac{1}{2}\mu_{P_{\psi}} + \frac{1}{2} \mu_{P_{\phi}}.
\end{equation}  
\end{assumption} 
For an ontic distribution $\mu_{P_{\psi}}(\lambda)$, the authors then define the following: 
\begin{eqnarray}
H_{\mathrm{\psi}}(\mu_{p_{\psi}}) := \int_{\Lambda} \mu_{P_{\psi}}(\lambda) \log(\mu_{P_{\psi}}(\lambda)) 
\end{eqnarray}
where $\log$ are is assumed to be base two. 
Calculating the entropy $H_{\mathrm{ontic}}$ of the ontological distribution $\mu_{P_{\rho}}$, one finds
\begin{equation} \label{eq:Shannon1}
\begin{aligned}
H_{\mathrm{ontic}}(\mu_{P_{\rho}}) & =-\int_{\Lambda} \mu_{P_{\rho}} \log \mu_{P_{\rho}} d \lambda \\
& =-\int_{U_\phi} \mu_{P_{\rho}} \log \mu_{P_{\rho}} d \lambda-\int_{U_\psi} \mu_{P_{\rho}} \log \mu_{P_{\rho}} d \lambda \\
& =-\int_{U_\psi} \frac{1}{2} \mu_{P_{\psi}} \log \frac{1}{2} \mu_{P_{\psi}} d \lambda-\int_{U_\phi} \frac{1}{2} \mu_{P_{\phi}} \log \frac{1}{2} \mu_{P_{\phi}} d \lambda \\
& =-\frac{1}{2} \int_{U_\psi} \mu_{P_{\psi}} \log \frac{1}{2} d \lambda-\frac{1}{2} \int_{U_\psi} \mu_{P_{\psi}} \log \mu_{P_{\psi}} d \lambda \\
& -\frac{1}{2} \int_{U_\phi} \mu_{P_{\phi}} \log \frac{1}{2} d \lambda-\frac{1}{2} \int_{U_\phi} \mu_{P_{\phi}} \log \mu_{P_{\phi}} d \lambda \\
& =-\frac{1}{2} \log \frac{1}{2}-\frac{1}{2} \log \frac{1}{2}+\frac{1}{2} H\left(\mu_{P_{\psi}}\right)+\frac{1}{2} H_{\mathrm{ontic}}\left(\mu_{P_{\phi}}\right) \\
& =1+\frac{1}{2} H_{\mathrm{ontic}}\left(\mu_{P_{\psi}}\right)+\frac{1}{2} H_{\mathrm{ontic}}\left(\mu_{P_{\phi}}\right) .
\end{aligned}
\end{equation} In the first and second line, we used that $\mu_{P_{\rho}}$ can be written as a convex combination of $\mu_{P_{\psi}}$ and $\mu_{P_{\phi}}$, which have nonoverlapping ontological support $U_\psi$ and $U_\phi$. Next, the authors assume that as pure states have zero entropy that this holds on the ontological level as well, i.e. \begin{equation}
H(\mu_{P_{\psi}}) = 0 = H(\mu_{P_{\phi}}).
\end{equation}  
\begin{assumption} \label{ass:pure_zero_entropy}
    For any pure quantum state $\psi$ we have $H_{\mathrm{ontic}}(\mu_{P_{\psi}})=0$. 
\end{assumption} 
As we argue below in \Cref{sec:psi-completeness} this implies $\psi$-completeness. We find using \cref{eq:Shannon1} that 
\begin{equation}
    H_{\mathrm{ontic}}(\mu_{P_{\rho}}) = 1,
\end{equation} which does not equal the von Neumann entropy for the quantum state $\rho$, given in \cref{eq:vonNeumann} as $|\psi\rangle$ and $|\phi\rangle$ are nonorthogonal, i.e. $p = |\langle \psi |\phi \rangle|^2 \neq 0$. The authors argue that this is contradicting, as they assume the von Neumann entropy of a quantum state should equal the ontological entropy, namely they assume the following. 
\begin{assumption} \label{ass:entropy_vN_equals_ontological}
    The entropy of a quantum state equals the entropy of its ontological representation. 
\end{assumption}

\vspace{0.5cm}

We can summarise these findings in the following theorem, the main result of \cite{carcassi2022reality} but now stating the (hidden) assumptions.

\begin{theorem} \label{th:carcassi2022reality}
    A $\psi$-complete model for quantum theory where for quantum states the ontological Shannon entropy equals the von Neumann entropy of quantum theory cannot be ontological for convexity-preserving preparations.
\end{theorem}

\section{Discussing the precise assumptions}
\subsection{The result of \cite{carcassi2022reality} assumes $\psi$-completeness instead of just $\psi$-onticness} \label{sec:psi-completeness}

It is straightforward to argue that assumption \ref{ass:pure_zero_entropy} directly entails $\psi$-completeness. 

\begin{lemma}
Assumption \ref{ass:convex-preserving} and Assumption \ref{ass:pure_zero_entropy} implies $\psi$ completeness.
\end{lemma}
\begin{proof}
Let $\psi$ be any pure state, and let $p_{\psi}$ denote its ontological distribution. Assumption \ref{ass:entropy_vN_equals_ontological} implies the following:

\begin{equation}\label{eqn: A2}
- \int_{\Lambda} d\lambda p(\lambda) \log_2(p(\lambda)) = 0
\end{equation}

The consequence of Equation \eqref{eqn: A2} is that for $p_{\psi} \in [0, 1]$, the expression $-p_{\psi} \log(p_{\psi}) > 0$ implies

\begin{eqnarray}
\forall \lambda \in \Lambda: \left( \quad p_{\psi}(\lambda) = 0 \qquad \text{or} \qquad p_{\psi}(\lambda) = 1 \right)
\end{eqnarray}

As $p_{\psi}$ is a probability distribution, this further implies that for each preparation $P_\psi$ there exists an ontic state $\lambda_{P_\psi}$ such that we have \begin{equation}
    \mu_{P_\psi}(\lambda) = \delta(\lambda-\lambda_{P_\psi}), 
\end{equation} but $\psi$-completeness needs \begin{equation}
    \mu_{P_\psi}(\lambda) = \delta(\lambda-\lambda_{\psi}).
\end{equation} 
If we additionally assume convexity-preserving preparations, i.e. \begin{equation}
    \frac{1}{2}\mu_{P_\psi}(\lambda)+\frac{1}{2}\mu_{P'_\psi}(\lambda)=\mu_{P''_\psi}(\lambda)
\end{equation} then we find that all these delta functions must be equal to each other, i.e. $\lambda_{P_\psi}=\lambda_\psi$. And if $H(\mu_{P_\psi})=0$, we recover the result of \cite{carcassi2022reality}. 
\end{proof}

Thus, any ontological model satisfying the (implicit) assumptions of \cite{carcassi2022reality}, namely \Cref{ass:convex-preserving} and \Cref{ass:pure_zero_entropy}, must be $\psi$-complete. 
This lemma thus shows that the claim of \cite{carcassi2022reality} already contradicts the fact that one can produce $\psi$-incomplete ontological models for quantum theory.

\subsection{The entropy of a quantum state and its underlying ontological distribution need not be equal}
Assumption \ref{ass:entropy_vN_equals_ontological} is not a good assumption for several reasons. Firstly, there is no inherent requirement that the von Neumann entropy of a state must be derived via the ontological model. From a purely information-theoretic perspective, most entropies are defined based on specific operational tasks and do not necessarily represent physical properties of the system. For instance, the operational meaning of von Neumann entropy is derived from scenarios like state discrimination \cite{wilde2013quantum}. Consider Alice generating an ensemble of pure states ${ \ket{\psi_{x}} }{x} $ with a probability distribution $p_{X}$ and sending it to Bob. Bob receives the state $\sigma = \sum_{x} p_{X}(x) \ket{\psi_{x}}\bra{\psi_{x}}$, and his task is to determine which state Alice has prepared. The von Neumann entropy $H(\sigma)$ quantifies Bob's uncertainty about the pure state prepared by Alice. However, this quantity is not inherently tied to the underlying "elements of reality" of the pure states themselves. As the task at hand is to distinguish pure states, it is reasonable to assume that if the ensemble consists of a single pure state $\ket{\psi}$, then the uncertainty is zero, implying $H(\ket{\psi}\bra{\psi}) = 0$. Within quantum theory, too, different entropic quantities are usually defined depending on the task at hand. The entropic quantity $H_{\mathrm{ontic}}$ would be the right operational quantity to consider if Alice prepares (if possible) an ensemble $\{ p(\lambda) , \lambda \}$ and Bob's task was to determine which $\lambda$ is sent.

Moreover, it is not the case that entropies from two different theoretical descriptions of the same information-theoretic task must lead to the same results. For instance, consider the state discrimination task for the ensemble $\{p_{X}(x) , \ket{\psi_{x}} \}$ as discussed earlier. This task can be approached via two different theoretical frameworks: one in classical theory and the other in quantum theory. Using the data-processing inequality, it becomes apparent that the Shannon entropy $S(X)$ (entropy for the classical case) cannot be smaller than the von Neumann entropy of the ensemble $H(\rho)$ (see Exercise 11.9.3 from \cite{wilde2013quantum}). In fact they differ from each other in most cases. Similarly, there is no inherent reason why the entropy $H_{\mathrm{ontic}}(.)$ of the ensemble should reproduce the von Neumann entropy $H(.)$, as they arise from different theoretical descriptions altogether. Thus, the quantum entropy $H(.)$ can differ from entropy $H_{\mathrm{ontic}}(.)$ even if they were to describe the same operational task.

Finally, regarding the Harrigan-Spekkens model, it is crucial to acknowledge that their description of quantum states is fundamental only to pure quantum states, i.e. mixed states need not have a unique ontological representation. The objective of the framework is to determine whether a quantum state is a state of complete knowledge. However, a mixed quantum state fails to meet this criterion, even within quantum theory. Mixed states are effective states arising themselves from a lack of classical knowledge. Hence, expecting the framework to establish a one-to-one correspondence with mixed states is an additional assumption proposed by the authors, which should not inherently hold. Nonetheless, this does not imply that mixed states are indescribable within the theory. For every mixed state, there exists a purifying system and a corresponding pure state. This pure state should be regarded as representing complete knowledge and can be in one-to-one correspondence with some probability $\mu$ on the ontic space. Demanding that every mixed state should also possess a preparation-independent distribution, as discussed further in section \ref{sec:contextuality}, is an additional assumption used on top of the Harrigan-Spekkens framework. As a result, even if all the aforementioned arguments are discounted, this is the sole reason why a one-to-one correspondence with mixed states is an added assumption and not an inherent one of the framework. Hence, it cannot serve as the basis for rejecting the framework.

\subsection{No fine-tuning no-go theorems rather than the nonexistence of ontological models for quantum theory} \label{sec:contextuality}

Furthermore, the statement `No ontological model can reproduce quantum theory' (\cite{carcassi2022reality}, p.14) is wrong, and there are examples of such models that reproduce quantum theory, such as the Beltrametti-Bugajski model \cite{beltrametti1995classical,spekkens2005contextuality}. 
However, a more correct statement that appears in quantum foundations is the following: \begin{quote}
    No ontological model that satisfies no-fine tuning can reproduce quantum theory.
\end{quote} Here \textit{no fine-tuning} means that operational identities are reflected in the ontological model \cite{catani2023mathematical,wood2015lesson,spekkens2005contextuality,spekkens2019ontological}, such as the same density operator having the same ontological representation, no matter how it is produced, also referred to as preparation noncontextuality \cite{spekkens2005contextuality}. Proofs of this result fall amongst others in the paradigm of generalised noncontextuality \cite{spekkens2005contextuality,spekkens2014status,schmid2018all,schmid2020structure,schmid2021characterization,chaturvedi2021characterising,mazurek2016experimental,chaturvedi2020quantum,selby2023contextuality,schmid2023addressing,pusey2018robust,khoshbin2023alternative}.

In fact, as we have established in \Cref{sec:psi-completeness}, the assumptions of \cite{carcassi2022reality} imply $\psi$-completeness and thus preparation noncontextuality for pure states, i.e. all pure states have the same ontological delta distribution.
Furthermore, \Cref{ass:entropy_vN_equals_ontological} implies that all ontological distributions of a mixed state preparation have the same entropy, they all have the same ontological shape, which gets close to assuming all preparations of the same mixed state to have the same ontological distribution, which amounts to assuming preparation noncontextuality for which many no-go theorems have already been established in the paradigm of generalised noncontextuality \cite{spekkens2005contextuality,spekkens2014status,schmid2018all,schmid2020structure,schmid2021characterization,chaturvedi2021characterising,mazurek2016experimental} without assuming any entropy relations.


\printbibliography

@book{wilde2013quantum,
    author = {Wilde, Mark M.},
    title = {Quantum Information Theory},
    publisher = {Cambridge University Press},
    year = {2013},
    onlinepublicationdate = {May 2013},
    printpublicationyear = {2013},
    onlineisbn = {9781139525343},
    doi = {10.1017/CBO9781139525343}
}

@article{carcassi2022reality,
  title={On the reality of the quantum state once again: A no-go theorem for $\psi $-ontic models},
  author={Carcassi, Gabriele and Oldofredi, Andrea and Aidala, Christine A},
  journal={arXiv preprint arXiv:2201.11842},
  year={2022}
}

@article{kunjwal2019beyond,
  title={Beyond the Cabello-Severini-Winter framework: Making sense of contextuality without sharpness of measurements},
  author={Kunjwal, Ravi},
  journal={Quantum},
  volume={3},
  pages={184},
  year={2019},
  publisher={Verein zur F{\"o}rderung des Open Access Publizierens in den Quantenwissenschaften}
}

@article{spekkens2005contextuality,
  title={Contextuality for preparations, transformations, and unsharp measurements},
  author={Spekkens, Robert W},
  journal={Physical Review A},
  volume={71},
  number={5},
  pages={052108},
  year={2005},
  publisher={APS}
}

@article{beltrametti1995classical,
  title={A classical extension of quantum mechanics},
  author={Beltrametti, EG and Bugajski, S},
  journal={Journal of Physics A: Mathematical and General},
  volume={28},
  number={12},
  pages={3329},
  year={1995},
  publisher={IOP Publishing}
}

@article{catani2023mathematical,
  title={A mathematical framework for operational fine tunings},
  author={Catani, Lorenzo and Leifer, Matthew},
  journal={Quantum},
  volume={7},
  pages={948},
  year={2023},
  publisher={Verein zur F{\"o}rderung des Open Access Publizierens in den Quantenwissenschaften}
}

@article{wood2015lesson,
  title={The lesson of causal discovery algorithms for quantum correlations: causal explanations of Bell-inequality violations require fine-tuning},
  author={Wood, Christopher J and Spekkens, Robert W},
  journal={New Journal of Physics},
  volume={17},
  number={3},
  pages={033002},
  year={2015},
  publisher={IOP Publishing}
}

@article{spekkens2019ontological,
  title={The ontological identity of empirical indiscernibles: Leibniz's methodological principle and its significance in the work of Einstein},
  author={Spekkens, Robert W},
  journal={arXiv preprint arXiv:1909.04628},
  year={2019}
}

@article{chaturvedi2020quantum,
  title={Quantum prescriptions are more ontologically distinct than they are operationally distinguishable},
  author={Chaturvedi, Anubhav and Saha, Debashis},
  journal={Quantum},
  volume={4},
  pages={345},
  year={2020},
  publisher={Verein zur F{\"o}rderung des Open Access Publizierens in den Quantenwissenschaften}
}

@article{selby2023contextuality,
  title={Contextuality without incompatibility},
  author={Selby, John H and Schmid, David and Wolfe, Elie and Sainz, Ana Bel{\'e}n and Kunjwal, Ravi and Spekkens, Robert W},
  journal={Physical Review Letters},
  volume={130},
  number={23},
  pages={230201},
  year={2023},
  publisher={APS}
}

@article{schmid2023addressing,
  title={Addressing some common objections to generalized noncontextuality},
  author={Schmid, David and Selby, John H and Spekkens, Robert W},
  journal={arXiv preprint arXiv:2302.07282},
  year={2023}
}

@article{schmid2021characterization,
  title={Characterization of noncontextuality in the framework of generalized probabilistic theories},
  author={Schmid, David and Selby, John H and Wolfe, Elie and Kunjwal, Ravi and Spekkens, Robert W},
  journal={PRX Quantum},
  volume={2},
  number={1},
  pages={010331},
  year={2021},
  publisher={APS}
}

@article{pusey2018robust,
  title={Robust preparation noncontextuality inequalities in the simplest scenario},
  author={Pusey, Matthew F},
  journal={Physical Review A},
  volume={98},
  number={2},
  pages={022112},
  year={2018},
  publisher={APS}
}

@article{spekkens2014status,
  title={The status of determinism in proofs of the impossibility of a noncontextual model of quantum theory},
  author={Spekkens, Robert W},
  journal={Foundations of Physics},
  volume={44},
  pages={1125--1155},
  year={2014},
  publisher={Springer}
}

@article{harrigan2010einstein,
  title={Einstein, incompleteness, and the epistemic view of quantum states},
  author={Harrigan, Nicholas and Spekkens, Robert W},
  journal={Foundations of Physics},
  volume={40},
  pages={125--157},
  year={2010},
  publisher={Springer}
}

@article{schmid2020structure,
  title={A structure theorem for generalized-noncontextual ontological models},
  author={Schmid, David and Selby, John H and Pusey, Matthew F and Spekkens, Robert W},
  journal={arXiv preprint arXiv:2005.07161},
  year={2020}
}

@article{chaturvedi2021characterising,
  title={Characterising and bounding the set of quantum behaviours in contextuality scenarios},
  author={Chaturvedi, Anubhav and Farkas, M{\'a}t{\'e} and Wright, Victoria J},
  journal={Quantum},
  volume={5},
  pages={484},
  year={2021},
  publisher={Verein zur F{\"o}rderung des Open Access Publizierens in den Quantenwissenschaften}
}

@article{mazurek2016experimental,
  title={An experimental test of noncontextuality without unphysical idealizations},
  author={Mazurek, Michael D and Pusey, Matthew F and Kunjwal, Ravi and Resch, Kevin J and Spekkens, Robert W},
  journal={Nature communications},
  volume={7},
  number={1},
  pages={ncomms11780},
  year={2016},
  publisher={Nature Publishing Group UK London}
}

@article{schmid2018all,
  title={All the noncontextuality inequalities for arbitrary prepare-and-measure experiments with respect to any fixed set of operational equivalences},
  author={Schmid, David and Spekkens, Robert W and Wolfe, Elie},
  journal={Physical Review A},
  volume={97},
  number={6},
  pages={062103},
  year={2018},
  publisher={APS}
}

@article{khoshbin2023alternative,
  title={Alternative robust ways of witnessing nonclassicality in the simplest scenario},
  author={Khoshbin, Massy and Catani, Lorenzo and Leifer, Matthew},
  journal={arXiv preprint arXiv:2311.13474},
  year={2023}
}

@article{pusey2012reality,
  title={On the reality of the quantum state},
  author={Pusey, Matthew F and Barrett, Jonathan and Rudolph, Terry},
  journal={Nature Physics},
  volume={8},
  number={6},
  pages={475--478},
  year={2012},
  publisher={Nature Publishing Group UK London}
}

\end{document}